\documentclass[fleqn,11pt]{article}
\usepackage{hyperref}
\usepackage{amsmath, amsthm,amssymb, multirow, graphicx,footnote,caption,bm,color}
\usepackage{color}
\newtheorem{algorithm}{Algorithm}[section]

\newtheorem{theorem}{Theorem}[section]

\usepackage{apacite}

\usepackage{amsmath,bm}
\usepackage{color}
\usepackage{rotating}

\usepackage{caption}
\usepackage{subcaption}

\usepackage{grffile}

\usepackage{graphicx}

\usepackage{lscape}

\usepackage{blindtext}

\usepackage{amssymb}
\usepackage{varioref} 
\usepackage{textcomp}

\usepackage{booktabs}

\usepackage{multirow}

\usepackage{pdflscape}

\usepackage{relsize}
\usepackage{tablefootnote}
\usepackage{wrapfig}
\usepackage{natbib}
\usepackage{listings}
\usepackage{longtable}
\usepackage{url}
\usepackage{tabularx}
\usepackage{float}
\usepackage{epsfig}
\usepackage[titletoc]{appendix}
\usepackage{arydshln}
\usepackage{footnote}
\usepackage{secdot}
\usepackage{bigints}
\usepackage[utf8]{inputenc}
\usepackage[english]{babel}
\usepackage{xr}
\usepackage{adjustbox}
\usepackage{graphicx}

\usepackage{multicol}
\usepackage{enumitem}

\usepackage{bigstrut}
\usepackage{rotating}
\usepackage{xr}

\newtheorem{corollary}{Corollary}[section]
\newcommand{\bT}{\boldsymbol\theta}

\begin{document}
\sloppy

\title{\bf{A Sequential Quadratic Hamiltonian-Based Estimation Method for Box-Cox Transformation Cure Model}}
\author{\bf{Phuong Bui${}^{1}$, Varun Jadhav${}^{1}$, Suvra Pal${}^{1,2,\footnote{Corresponding author. Department of Mathematics, University of Texas at Arlington, 411 S. Nedderman Dr., Arlington, TX 76019, United States. E-mail address: suvra.pal@uta.edu Tel.:817-272-7163.\newline \indent }}$, and Souvik Roy${}^{1}$}  \\\\
${}^{1}$Department of Mathematics, University of Texas at Arlington, \\
Arlington, TX 76019, United States.\\
${}^{2}$Division of Data Science, College of Science, University of Texas at Arlington, \\
Arlington, TX 76019, United States}
\date{}

\maketitle

\begin{abstract}

\noindent We propose an enhanced estimation method for the Box-Cox transformation (BCT) cure rate model parameters by introducing a generic maximum likelihood estimation algorithm, the sequential quadratic Hamiltonian (SQH) scheme, which is based on a gradient-free approach. We apply the SQH algorithm to the BCT cure model and, through an extensive simulation study, compare its model fitting results with those obtained using the recently developed non-linear conjugate gradient (NCG) algorithm. Since the NCG method has already been shown to outperform the well-known expectation maximization algorithm, our focus is on demonstrating the superiority of the SQH algorithm over NCG. First, we show that the SQH algorithm produces estimates with smaller bias and root mean square error for all BCT cure model parameters, resulting in more accurate and precise cure rate estimates. We then demonstrate that, being gradient-free, the SQH algorithm requires less CPU time to generate estimates compared to the NCG algorithm, which only computes the gradient and not the Hessian. These advantages make the SQH algorithm the preferred estimation method over the NCG method for the BCT cure model. Finally, we apply the SQH algorithm to analyze a well-known melanoma dataset and present the results.

\end{abstract}


\section{Introduction}

Recent advancements in treating diseases like cancer and heart disease have extended the lives of many patients. Among these, a group known as recurrence-free survivors includes those who respond well to treatment and remain in long-term remission. A subset of these individuals, whose disease is no longer detectable and poses no harm, are often referred to as long-term survivors or “cured". However, determining which of these recurrence-free survivors can truly be considered cured presents a challenge, especially when some patients who might still be at risk of disease recurrence live through the duration of the study measuring survival rates. The aim of “cure rate" modeling \citep{PengYu21} is to resolve this uncertainty by providing an unbiased estimate of the proportion of patients who are truly cured. Accurately estimating the cure rate for a specific treatment is crucial for tracking patient survival trends and is clinically important, as it helps identify patients who can avoid the additional risks associated with high-intensity treatments.

The development of cure rate models traces back to the pioneering work of \cite{Boa49}, which was later refined by \cite{Ber52}. This model, known as the mixture cure rate model (MCM), defines the population survival function (or long-term survival function) for a time-to-event variable $Y$ as:
\begin{equation}
S_p(y) = p_0 + (1-p_0)S_s(y),
\label{mix}
\end{equation}
where $p_0$ represents the proportion of cured individuals (or the cure rate) and $S_s(\cdot)$ is the survival function of the susceptible group (a proper survival function). Although MCM-based models have been widely applied in cancer research \citep{Kim09,Lam10}, cardiovascular diseases \citep{Sev21}, and infectious diseases \citep{Ped22}, they fall short in addressing practical situations where multiple latent causes contribute to the occurrence of an event. For example, various malignant cells contribute to the development of a cancerous tumor. As a result, these models lack clinical relevance since they fail to accurately represent the biological mechanisms behind an event's occurrence.

To address this limitation, promotion time cure model (PCM) was introduced, which incorporates the progression times of latent causes \citep{Che99,Tso03}. In the PCM model, the number of latent causes is assumed to follow a Poisson distribution, and the long-term survival function is expressed as:
\begin{equation}
S_p(y) = e^{-\eta\{1-S(y)\}},
\label{prom}
\end{equation}
where $\eta$ represents the mean number of causes and $S(\cdot)$ is the common survival function of the progression time for each cause. In this model, the cure rate is expressed as $e^{-\eta}$. In real-world scenarios, the cure rate can vary across individuals, so it must depend on subject-specific characteristics or covariates. Let $\bm x$ represent the covariate vector associated with the cure rate. To analyze the effect of covariates on the cure rate, we can relate $p_0$ in the MCM model to $\bm x$ via a logistic link function:
\begin{equation}
p_0 = p_0(\bm x) = \frac{1}{1+\exp(\bm x^\prime \bm \beta)},
\end{equation}
where $\bm\beta$ is the vector of regression coefficients. Similarly, in the PCM model, we can link $\eta$ to $\bm x$ through a log-linear function:
\begin{equation}
\eta= \eta(\bm x) = \exp(\bm x^\prime \bm \beta).
\end{equation}
In this context, readers interested in more advanced methods can explore recent machine learning-based approaches for modeling the cure rate to enhance predictive accuracy \citep{Ase23,PalSVM23,PalSMMR23,Paletal23}. Note that the PCM model struggles with over- and under-dispersion, common issues when modeling count data. To address this, the Conway-Maxwell Poisson (COM-Poisson) cure model was developed, offering a flexible framework for dealing with both over- and under-dispersion relative to the Poisson distribution \citep{Bal12,Bal13,Bal15,Bal16,Pal17b}. Additionally, as an extension of cure models in the latent cause framework, researchers have explored scenarios where causes may be eliminated or destroyed after a certain period (such as following an initial treatment phase). Several destructive cure models have been developed, along with likelihood-based inference procedures to handle these cases \citep{Pal16,Pal17c,Pal17a,JodiSIM23,Jodi22}.

In this paper, we explore the unification of the MCM and PCM models through a class of cure models based on the Box-Cox transformation (BCT) applied to the population survival function. This class was first proposed by \cite{Yin05}. The BCT cure rate model was recently investigated by \cite{PalRoy23}, where the authors introduced a new estimation method using a non-linear conjugate gradient (NCG) algorithm with an efficient line search technique. This NCG-based approach enabled the simultaneous maximization of all parameters, including the BCT index parameter $\alpha$. This contrasts with the expectation maximization (EM) algorithm developed by \cite{Pal18c} for the BCT cure model, where simultaneous maximization of all parameters was not feasible, and $\alpha$ was estimated using a profile likelihood approach within the EM algorithm. In the original work of \cite{Yin05}, the parameter $\alpha$ was fixed, and Bayesian inference was used. In their development of the NCG algorithm, \cite{PalRoy23} demonstrated that the NCG method resulted in lower bias and significantly reduced root mean square error in the parameter estimates associated with the cure rate, when compared to the EM algorithm. This leads to more accurate and precise inference on the cure rate, and the NCG method was also found to be computationally more efficient than the EM algorithm \citep{PalRoy21,PalRoy22}. 

We propose a novel estimation method for the BCT cure model, called the sequential quadratic Hamiltonian (SQH) method. The SQH method is based on the Pontryagin's maximum principle (PMP) characterization and is robust and fast converging. Since the PMP principle has a pointwise formulation and, thus, a local structure, the SQH algorithm implements local updates of the unknown parameters. It does not require gradient computation and is expected to be even more computationally efficient. The SQH scheme has been used in past in the context of a Liouville mass transport problem \citep{Roy_Liouv}, medical imaging problems \citep{Dey2024,Roy2021,Roy2023} and parabolic optimal control problems \citep{Breitenbach2019a}, and has been demonstrated to be effective optimization solvers. Since the NCG method has already been shown to outperform the EM algorithm, our focus is on demonstrating the superiority of the proposed SQH algorithm over the NCG. We aim to highlight the advantages of the SQH method and believe it will encourage researchers to adopt it as a new tool for parameter estimation, both in the context of cure models and in other applications. To the best of our knowledge, this is the first time an SQH-based method has been developed in the context of cure models.

The rest of this paper is organized as follows. Section 2 introduces the BCT cure rate model. In Section 3, we describe the observed data structure and the likelihood function, and outline the steps of the SQH algorithm as a general method for estimation. We also examine some of its theoretical properties. Section 4 focuses on using the SQH algorithm to estimate the parameters of the BCT cure model, presenting results from an extensive simulation study. We demonstrate the accuracy of the SQH algorithm in estimating all parameters and emphasize its advantages over the NCG algorithm. Section 5 applies the SQH algorithm to analyze a well-known melanoma dataset. Finally, Section 6 concludes the paper with some remarks and explores potential avenues for future research.

\section{BCT cure rate model}

The BCT applied to a variable $z$ with index parameter $\alpha$ is defined as:
\begin{equation}
G(z,\alpha)=\begin{cases}
\frac{z^\alpha -1}{\alpha}, & \alpha \neq 0 \\
\log(z), & \alpha = 0.
\end{cases}
\label{BC} 
\end{equation}
Let the population survival function depend on a set of covariates $\boldsymbol{x}$, denoted by $S_p(\cdot|\boldsymbol{x})$. Applying the BCT to the population survival function, as suggested by \cite{Yin05}, the BCT cure rate model is given by:
\begin{equation}
G(S_p(y|\boldsymbol x),\alpha) = -\phi(\alpha,\boldsymbol x)F(y), \ \ 0\leq\alpha \leq 1.
\label{G}
\end{equation}
In equation \eqref{G}, $y$ represents the observed lifetime, $\alpha$ is the index parameter, $F(\cdot)$ is any valid distribution function, and $\boldsymbol{x}$ represents a set of covariates introduced in the model via a general covariate structure $\phi(\alpha,\boldsymbol{x})$, defined as:
\begin{equation}
\phi(\alpha,\boldsymbol x) = \begin{cases}
\frac{\exp(\boldsymbol x^\prime\boldsymbol\beta)}{1+\alpha\exp(\boldsymbol x^\prime\boldsymbol\beta)}, & \ \ \ \ 0<\alpha\leq 1 \\ 
\exp(\boldsymbol x^\prime\boldsymbol\beta), & \ \ \ \ \alpha = 0.
\end{cases}
\label{phi}
\end{equation}
Here, $\boldsymbol{\beta}$ denotes the set of regression coefficients. Substituting the BCT from equation \eqref{BC} into \eqref{G}, we obtain an expression for the population survival function as:
\begin{equation}
S_p(y|\boldsymbol x) = \begin{cases}
\{1-\alpha\phi(\alpha,\boldsymbol x)F(y)\}^{\frac{1}{\alpha}}, & 0 < \alpha \leq 1\\
\exp\{-\phi(0,\boldsymbol x)F(y)\}, & \alpha = 0.
\end{cases}
\label{Sp}
\end{equation}
The cure rate, which represents the long-term survival probability, is defined as $p_0(\boldsymbol{x}) = \lim_{y \to \infty} S_p(y|\boldsymbol{x})$. Its explicit expression is:
\begin{eqnarray}
p_0(\boldsymbol x)&=&\begin{cases}
\{1-\alpha\phi(\alpha,\boldsymbol x)\}^{\frac{1}{\alpha}}, & 0 < \alpha \leq 1\\
\exp\{-\phi(0,\boldsymbol x)\}, & \alpha = 0.
\end{cases}
\label{p0}
\end{eqnarray}
The population density function can be derived as $f_p(y|\boldsymbol{x}) = S_p'(y|\boldsymbol{x})$, and its explicit form is:
\begin{eqnarray}
f_p(y|\boldsymbol x) = \begin{cases}
S_p(y|\boldsymbol x)\phi(\alpha,\boldsymbol x)f(y)\{1-\alpha\phi(\alpha,\boldsymbol x)F(y)\}^{-1}, & 0 < \alpha \leq 1\\ 
S_p(y|\boldsymbol x)\phi(0,\boldsymbol x)f(y), & \alpha = 0,
\end{cases}
\label{fp}
\end{eqnarray} 
where $f(\cdot)$ is the density function corresponding to $F(\cdot)$. Notably, the BCT cure model in equation \eqref{Sp} reduces to the MCM in \eqref{mix} when $p_0 = p_0(\boldsymbol{x}) = \frac{1}{1 + \exp(\boldsymbol{x}^\prime\boldsymbol\beta)}$ and $S_s(y) = 1 - F(y)$. Additionally, the BCT cure model simplifies to the PCM in \eqref{prom} with $S(y) = 1 - F(y)$ and $\eta = \exp(\boldsymbol{x}^\prime\boldsymbol\beta)$. Therefore, the BCT cure model unifies two of the most commonly used cure rate models in the literature, offering a flexible class of transformation cure models. We focus on the range of $\alpha$ within the interval $[0,1]$, as it provides an intermediate modeling structure between the PCM ($\alpha = 0$) and MCM ($\alpha = 1$) models.

While the MCM and PCM models are widely used and extensively studied in the literature, the BCT cure model has not been explored in depth. Only a few studies have examined the BCT cure model since the original work of \cite{Yin05}. Notably, \cite{YinIb05} introduced a new class of survival regression models that connect a family of both improper and proper population survival functions through the BCT applied to the population hazard function and a proper density function. \cite{Zen06} developed a class of transformation cure models, encompassing proportional hazards and proportional odds cure models, and proposed a recursive algorithm for parameter estimation. \cite{Sunetal11} extended the BCT cure model to recurrent event data and developed a profile pseudo-partial likelihood method for parameter estimation. \cite{Peng-Xu12} provided the first novel biological interpretation of the BCT cure model. \cite{Choi12} introduced a systematic martingale approach for determining the index parameter in a semi-parametric transformation cure model. \cite{Kou17} proposed a family of transformation cure models that includes both classical and destructive cure models. As previously mentioned, \cite{Pal18c} considered the BCT cure model in a fully parametric framework and developed likelihood inference using the EM algorithm. \cite{Wang22} proposed an extended class of generalized gamma BCT cure models, addressing issues of model selection and discrimination. \cite{Mil22} introduced a new reparameterization for a family of transformation cure models, where the case of a zero cured proportion is not constrained at the boundary of the parameter space. More recently, \cite{JodiCS24} developed likelihood-based inference for the BCT cure model with interval-censored data; see also \cite{PalBarui24}.

\section{Estimation method: SQH algorithm}

\subsection{Form of the data and likelihood function}

We consider a scenario where time-to-event (or lifetime) data may be subject to right censoring. Let $T_i$ and $C_i$ represent the actual failure time and censoring time, respectively, for $i = 1, 2, \dots, n$, where $n$ denotes the sample size. The observed lifetime is then given by $Y_i = \min\{T_i, C_i\}$. Let $\delta_i$ be the right censoring indicator, taking the value 1 if the event time is observed and 0 if it is right censored. The observed data set is thus represented by $\boldsymbol{O} = \{(y_i, \delta_i, \boldsymbol{x}_i), i = 1, 2, \dots, n\}$, where $\boldsymbol{x}_i$ is the vector of covariates for the $i$-th subject. Assuming non-informative censoring, the observed data likelihood function can be written as:
\begin{eqnarray}
L(\boldsymbol\theta)= \prod_{i=1}^n \{f_p(y_i|\boldsymbol {x_i})\}^{\delta_i}\{S_p(y_i|\boldsymbol{x_i})\}^{1-\delta_i}, \label{eq:likelihood}
\end{eqnarray}
where $\boldsymbol{\theta} = (\boldsymbol{\beta}^\prime, \boldsymbol{\gamma}^\prime,\alpha)^\prime$ is the vector of unknown model parameters, with $\boldsymbol{\gamma}$ representing the parameter vector associated with $F(\cdot)$ in equation \eqref{Sp}. To estimate the optimal parameter set $\boldsymbol{\theta}$ that maximizes the likelihood function in \eqref{eq:likelihood}, we first take the natural logarithm of both sides of \eqref{eq:likelihood} to obtain the log-likelihood function:
\begin{equation}\label{eq:loglikelihood}
l(\boldsymbol\theta)= \sum_{i=1}^n[ \delta_i \log\{f_p(y_i|\boldsymbol {x_i})\}+{(1-\delta_i)}\log\{S_p(y_i|\boldsymbol{x_i})\}],
\end{equation}
where $S_p(\cdot|\boldsymbol{x}_i)$ and $f_p(\cdot|\boldsymbol{x}_i)$ are as defined in equations \eqref{Sp} and \eqref{fp}, respectively.

\subsection{Proposed SQH algorithm}

To obtain the unknown model parameter set $\bT$, we use the SQH method. In this regard, the maximization problem can be stated as: $$\max_{\boldsymbol\theta\in T_{ad}}~ l(\boldsymbol\theta),$$ where $T_{ad}$ is the set of admissible values of $\bm\theta$, defined as: 

$$T_{ad} = \lbrace \boldsymbol\theta : \boldsymbol\beta\in\mathbb{R}^p,\boldsymbol\gamma\in (\mathbb{R^+})^q, ~ 0\leq \alpha \leq 1 \rbrace.$$ Thus, in general, the dimension of $\bm \theta$ is $p+q+1$. The starting point of our SQH method is the augmented Hamilton function, defined as follows:
$$l_\epsilon(\bT,\tilde{\bT}) = l(\bT) - \epsilon \|\bT - \tilde{\bT}\|^2_{l^2}.
\label{eq:augmented_Hamiltonian}
$$
Here $\epsilon > 0$ is a penalization parameter that is adaptively adjusted at each iteration of the SQH process, $\|\cdot\|_{l^2}$ is the standard vector Euclidean norm, and $\tilde{\bT}$ represents the previous approximations of the parameter set $\bT$. Specifically, $\epsilon$ is increased when a sufficient increase in the function $l(\cdot)$ is not observed, and decreased when $l(\cdot)$ increases adequately. The purpose of the quadratic term $\epsilon \|\bT - \tilde{\bT}\|^2_{l^2} $ is to ensure that the maximizer of $l_\epsilon$ for each individual parameter in $\bT$ remain close to the prior values $\tilde{\bT}$, especially when $\epsilon$ is large. The SQH algorithm is outlined below:

\begin{algorithm}[SQH scheme]\label{eq:SQH}\
\begin{enumerate}
\item  Choose $\epsilon > 0, \kappa >0, \lambda > 1$ (step-up parameter), $\zeta \in (0,1)$ (step-down parameter), $\rho \in (0,\infty)$, $\bT^0 \in T_{ad}$ (Initial guess), and $ ~MaxIter$ (denoting the maximum number of iterations).
\item Set $k = 0$.
\item Choose $\theta_i\in T_{i,ad}$, where $T_{i,ad}$ denotes the admissible set for the $i$-th $(i=1,2,\cdots,p+q+1)$ element in $\bm\theta$, such that 
\begin{eqnarray*}
l_\epsilon((\theta_1,\theta_2^k,\cdots,\theta_{p+q+1}^k),\bT^k) &=&  \mathop{\max_{v_1\in T_{1,ad}}} l_\epsilon((v_1,\theta_2^k,\cdots,\theta_{p+q+1}^k),\bT^k)\\
l_\epsilon((\theta_1^k,\theta_2,\cdots,\theta_{p+q+1}^k),\bT^k) &=&  \mathop{\max_{v_2\in T_{2,ad}}} l_\epsilon((\theta_1^k,v_2,\cdots,\theta_{p+q+1}^k),\bT^k),\\
& \vdots \\
l_\epsilon((\theta_1^k,\theta_2^k,\cdots,\theta_{p+q+1}),\bT^k) &=&  \mathop{\max_{v_{p+q+1}\in T_{p+q+1,ad}}} l_\epsilon((\theta_1^k,\theta_2^k,\cdots,v_{p+q+1}),\bT^k).
\end{eqnarray*}
\item Set $\tau = \|\bT-\bT^k\|^2_{l_2}.$
\item If $l(\bT) - l(\bT^k) < \rho \tau$ (Not sufficient increase), choose $\epsilon = \lambda \epsilon$ and go back to Step 3 to recompute $\bT$.
\item If $l(\bT) - l(\bT^k) \geq  \rho \tau$ (Sufficient increase), choose $\epsilon = \zeta \epsilon$ and assign new iterate $\bT^{k+1} = \bT$ obtained from Step 3.
\item Set $k = k+1$
\item If $\tau < \kappa$ or $k > MaxIter$, STOP and return optimal $\bT^k$. Else, go to Step 3.   
\item end
\end{enumerate}
\end{algorithm}

\begin{theorem}
\label{theoremSQH}
Let $\left(\bT^k\right)$ and $\left(\bT^{k+1}\right)$
be generated by the SQH method outlined in Algorithm \ref{eq:SQH}, where both $\bT^{k+1}$ and $\bT^k \in T_{ad}$. Then, the following holds: 
\begin{equation}\label{eq:ineq}
l(\bT^{k+1})-l(\bT^k)\geq\epsilon
\,   \|  \bT^{k+1}-\bT^k \|^{2}_{l^2}.
\end{equation}
\end{theorem}
\begin{proof}
After Step 6 of Algorithm \ref{eq:SQH}, we have
$$
l_\epsilon(\bT^{k+1},\bT^k)  \geq l_\epsilon(\bT^{k},\bT^k) = l(\bm\theta^k).
$$
This implies
$$
l(\bT^{k+1}) - \epsilon\|\bT^{k+1} - \bT^k\|^2_{l^2} \geq l(\bT^k).
$$
Thus, we have
$$
l(\bT^{k+1}) - l(\bT^k) \geq  \epsilon\|\bT^{k+1} - \bT^k\|^2_{l^2}.
$$
\end{proof}

\begin{corollary} \label{CorollaryEps}
The sequence $(\epsilon)$ of the SQH iterates is bounded.
\end{corollary}
\begin{proof}
From Algorithm \ref{eq:SQH}, we see that the successful 
$k$-th maximization step is performed with $\epsilon= \, \rho$. Then there is a constant $\bar{\epsilon} $ depending on 
the data, and the optimization and SQH parameters, which 
provides an upper bound $\epsilon \le \bar{\epsilon}$ of the sequence $(\epsilon)$ of the SQH iterates.
\end{proof}

\begin{theorem}\label{Tsqh2}
Let the sequence $(\bT^k)$ be generated by Algorithm \ref{eq:SQH}. Then, the sequence of  function values $l(\bT^k)$ monotonically increases with 
$$
\lim_{k\to\infty}\left[l(\bT^{k+1})-l(\bT^k)\right]=0 \ \ \text{and} \ \ \lim_{k\to\infty}\left\Vert \bT^{k+1}-\bT^k\right\Vert_{l^2}=0.
$$ 
\end{theorem}
\begin{proof}
From \eqref{eq:ineq}, choosing $\epsilon = \rho$, we obtain, 
$$
\|\bT^{k+1} - \bT^k\|^2_{l^2} \leq \dfrac{1}{\epsilon} \left[l(\bT^{k+1}) - l(\bT^k)\right] = \dfrac{1}{\rho} \left[l(\bT^{k+1}) - l(\bT^k)\right].
$$
This implies
$$
\sum_{k=0}^K \|\bT^{k+1} - \bT^k\|^2_{l^2} \leq \dfrac{1}{\rho} \left[l(\bT^{0}) - l(\bT^K)\right],
$$
which means that the series $\sum_{k=0}^\infty \|\bT^{k+1} - \bT^k\|^2_{l^2}$ is convergent. Thus, 
$$
\lim_{k\rightarrow \infty} \|\bT^{k+1} - \bT^k\|^2_{l^2} = 0.
$$
Since, 
$$
l(\bT^{k+1}) - l(\bT^k) \geq  \epsilon\|\bT^{k+1} - \bT^k\|^2_{l^2} \geq 0,
$$
we have 
$$
l(\bT^{k+1}) \geq l(\bT^k),~ \forall k \in \mathbb{N},
$$
which implies that the sequence $l(\bT^k)$ is monotonically increasing. Since it is bounded above, the sequence converges. As a result, we have
$$
\lim_{k\to\infty}\left[l(\bT^{k+1})-l(\bT^k)\right]=0.
$$
\end{proof}
Theorem \ref{Tsqh2} guarantees that Algorithm \ref{eq:SQH} is well defined for $\kappa>0$. Hence, there is an iteration number $k_0\in\mathbb{N}$ such that $\| \bT^{k_0+1}-\bT^{k_0}\|_{l^2}\leq\kappa$ and therefore the 
SQH algorithm stops in finitely many steps.

\section{Simulation study}

To assess the performance of the proposed SQH algorithm, we conducted a detailed Monte Carlo simulation study. For this, we set $tol = 0.001$ and $MaxIter=1000$. Regarding the fixed parameters associated with the development of the SQH algorithm, and after a careful preliminary study, we chose $\epsilon=1000$, $\lambda=1000$, $\rho=1000$, and $\zeta=0.5$. To compare the SQH algorithm's performance with the previously developed NCG algorithm as outlined by \cite{PalRoy23}, we used the same simulation setup as in their study. A brief description of the simulation setup is provided below. For the sake of comparison, and as done in \cite{PalRoy23}, we assume that $F(\cdot)$ in equation \eqref{Sp} and $f(\cdot)$ in equation \eqref{fp} correspond to the distribution function and density function of a two-parameter Weibull distribution, respectively. Thus, we have:
\begin{eqnarray}
F(y)&=&1-\exp\{-(\gamma_2 y)^{\frac{1}{\gamma_1}}\} \ \ \text{and} \nonumber \\ 
f(y) &=& \frac{1}{\gamma_1 y}(\gamma_2 y)^{\frac{1}{\gamma_1}}\{1-F(y)\} 
\label{wei}
\end{eqnarray}
for $y > 0$, $\gamma_1 > 0$, and $\gamma_2 > 0$. Additionally, $\boldsymbol{\gamma} = (\gamma_1, \gamma_2)^\prime$. It is important to note that the SQH algorithm can be easily extended to accommodate other semi-parametric or non-parametric methods for estimating $F(\cdot)$ and $f(\cdot)$, as well as to include covariates in $F(\cdot)$ or $f(\cdot)$.

\subsection{Presence of binary covariate} \label{sec:SA}

In this scenario, we consider a binary covariate $x$ (e.g., gender, presence or absence of an ulcer, etc.). This corresponds to a situation where patients are divided into two groups, with $x=1$ representing those in group 1 and $x=0$ representing those in group 2. The cured proportions for groups 1 and 2 are denoted by $p_{01}$ and $p_{00}$, respectively. Assuming the true values for the cured proportions $p_{01}$ and $p_{00}$ and a chosen value for $\alpha$, we can compute the true values of the regression parameters as follows:
 \begin{eqnarray*}
\beta_0& =& \begin{cases}
\log\bigg\{\frac{1}{\alpha}\bigg(\frac{1}{p_{00}^\alpha}-1\bigg)\bigg\}, & 0 < \alpha \leq 1\\ 
\log\{-\log(p_{00})\}, & \alpha = 0 
\end{cases}
\end{eqnarray*} 
and
\begin{eqnarray*}
\beta_1&=&\begin{cases}
\log\bigg\{\frac{1}{\alpha}\bigg(\frac{1}{p^\alpha_{01}}-1\bigg)\bigg\}-\beta_0, & 0 < \alpha \leq 1\\ 
\log\{-\log(p_{01})\}-\beta_0, & \alpha = 0.
\end{cases}
\end{eqnarray*}

For a given group, let the cured proportion be denoted by $p_0$. The following steps can be used to generate data from the BCT cure model in equation \eqref{Sp}:
\begin{itemize} 
\item[(i)] First, generate a random variable $U$ from a Uniform(0,1) distribution and a censoring time $C$ from an exponential distribution with rate $c$. Note that the censoring rate may differ between groups; 
\item[(ii)] If $U \leq p_0$, set the observed time $Y$ equal to the censoring time, i.e., $Y = C$; \item[(iii)] If $U > p_0$, equate the survival probability of susceptible patients, $\frac{S_p(y|\boldsymbol{x}) - p_0}{1 - p_0}$, to a new random variable $U^* \sim \text{Uniform}(0,1)$. This leads to the following equation:
\begin{eqnarray*}
&&
\begin{cases}
\frac{\{1-\alpha\phi(\alpha, x)F(t)\}^{\frac{1}{\alpha}}-p_0}{1-p_0}= U^*, & 0 < \alpha \leq 1  \\\\
\frac{\exp\{-\phi(0,x)F(t)\}-p_0}{1-p_0}= U^*, & \alpha = 0
\end{cases}\\
&\Rightarrow&
\begin{cases}
t=F^{-1}\bigg[\frac{1-\{p_0+(1-p_0)U^*\}^\alpha}{\alpha\phi(\alpha, x)}\bigg], & 0 < \alpha \leq 1  \\\\
t=F^{-1}\bigg[\frac{-\log\{p_0+(1-p_0)U^*\}}{\phi(0, x)}\bigg], & \alpha = 0,
\end{cases}
\end{eqnarray*}
where $F$ is the distribution function of the Weibull distribution, as defined in equation \eqref{wei}. The observed time can then be obtained as $Y = \min\{t, C\}$; 
\item[(iv)] Based on (ii) and (iii), if $Y = C$, set $\delta = 0$; otherwise, set $\delta = 1$.
\end{itemize}

We adopt the same parameter settings as in \cite{PalRoy23} (see also \cite{Pal18c}), which include three different sample sizes: $n=150$ ($n_1=75$, $n_2=75$), $n=200$ ($n_1=120$, $n_2=80$), and $n=300$ ($n_1=180$, $n_2=120$). We also consider two different sets of cured proportions $(p_{01}, p_{00})$: $(0.65, 0.35)$ and $(0.40, 0.20)$. Additionally, we set $(\gamma_1, \gamma_2)$ to $(0.316, 0.179)$ and the censoring rates for groups 1 and 2 to $(0.15, 0.10)$. For determining the initial values of the model parameters, we follow the same procedure outlined in \cite{Pal18c}.

\subsection{Presence of continuous covariate}

In this scenario, we consider a continuous covariate $x$ (e.g., tumor thickness measured in mm), which is generated from a Uniform(0.1, 20) distribution. Since we expect the cured proportion to decrease as tumor thickness increases, we set a low cured proportion of $p_{\text{low}} = 5\%$ for a patient with tumor thickness $x_{\text{max}} = 20$mm, and a high cured proportion of $p_{\text{high}} = 65\%$ for a patient with tumor thickness $x_{\text{min}} = 0.1$mm. Using these values and a fixed choice of $\alpha$, we can compute the true regression parameters as follows:
 \begin{eqnarray*}
\beta_1& =& \begin{cases}
\frac{\log\bigg\{\frac{1}{\alpha}\bigg(\frac{1}{p_{\text{low}}^\alpha}-1\bigg)\bigg\} - \log\bigg\{\frac{1}{\alpha}\bigg(\frac{1}{p_{\text{high}}^\alpha}-1\bigg)\bigg\}}{x_{max}-x_{min}}, & 0 < \alpha \leq 1\\\\
\frac{\log\{-\log(p_{\text{low}})\}-\log\{-\log(p_{\text{high}})\}}{x_{max}-x_{min}}, & \alpha = 0 
\end{cases}
\end{eqnarray*} 
and
\begin{eqnarray*}
\beta_0&=&\begin{cases}
\log\bigg\{\frac{1}{\alpha}\bigg(\frac{1}{p^\alpha_{\text{low}}}-1\bigg)\bigg\}-\beta_1x_{max}, & 0 < \alpha \leq 1\\ 
\log\{-\log(p_{\text{low}})\}-\beta_1x_{max}, & \alpha = 0.
\end{cases}
\end{eqnarray*}

For a patient with tumor thickness $x$ and a chosen true value of $\alpha$, the following steps are used to generate data from the BCT cure model in~\eqref{Sp}:

\begin{itemize}
\item[(i)] First, we calculate the cure rate as: 
\begin{eqnarray*}
p_0(x)&=&
\begin{cases}
\bigg\{\frac{1}{1+\alpha\exp(\beta_0+\beta_1x)}\bigg\}^{\frac{1}{\alpha}}, & 0 < \alpha \leq 1\\
\exp\{-\exp(\beta_0+\beta_1x)\}, & \alpha = 0;
\end{cases}
\end{eqnarray*}
\item[(ii)] Next, we generate a random variable $U^* \sim$ Uniform(0,1) and a censoring time $C$ from an exponential distribution with rate $c$; 
\item[(iii)] If $U^* \leq p_0(x)$, the observed time $Y$ is set to the censoring time, i.e., $Y = C$; 
\item[(iv)] If $U^* > p_0(x)$, we proceed as in step (iii) of section \ref{sec:SA}, replacing $p_0$ with $p_0(x)$; 
\item[(v)] Based on steps (iii) and (iv), if $Y = C$, we set $\delta = 0$; otherwise, we set $\delta = 1$. 
\end{itemize}

In this study, we use two different sample sizes: $n = 150$ and $n = 300$. The true values of the lifetime parameters and the method for finding initial values remain the same as in the simulation setup with the binary covariate.

\subsection{Simulation results}

Tables \ref{table:T1} and \ref{table:T2} present the simulation results for the BCT cure model with 
$\alpha=0.5$ and $\alpha=0.75$, respectively, with a binary covariate, in terms of absolute bias and RMSE. The SQH algorithm accurately converges to the true parameter values with small biases and RMSEs. Generally, both bias and RMSE decrease as the sample size increases, which is a desirable property. Compared to the NCG algorithm, the SQH algorithm consistently reduces bias and RMSE for all model parameters, highlighting its advantage. In Table \ref{table:T3}, we report the biases and RMSEs for cases where the index parameter $\alpha$ is on the boundary (i.e., $\alpha=0$ or 1). For other parameter settings, the results are similar and not presented for brevity. Even in the boundary cases, the SQH algorithm yields smaller biases and RMSEs than the NCG algorithm, further demonstrating its superiority. Table \ref{table:T4} presents results for a continuous covariate, with findings consistent with those in Tables \ref{table:T1}-\ref{table:T3}.

In Table \ref{table:Sp}, we present the biases and RMSEs of the estimates for the population survival function $S_p(y|x)$ in the general BCT model for $y=2$ and a specific value of the covariate $x$. Similar results for other values of $y$ and $x$, as well as for other parameter settings, are not reported for brevity. The results in Table \ref{table:Sp} again show that the SQH algorithm yields smaller biases and RMSEs for the estimated population survival function compared to the NCG algorithm. Notably, the biases and RMSEs produced by the NCG algorithm are very close to those from SQH.

We now focus on the bias and RMSE of the cure rates, presented in Table \ref{table:cure} for the general BCT cure model with a binary covariate. For the continuous covariate, the cure rate will vary for each subject, so it is not included here. Across all parameter settings, the SQH algorithm provides more accurate (i.e., smaller bias) and precise (i.e., smaller RMSE) estimates of cure rates compared to the NCG algorithm. Such accuracy and precision are clinically significant, as patients with high cure rates can avoid unnecessary high-intensity treatments, while those with low cure rates can receive timely interventions before the disease progresses to a more advanced stage, where treatment options are limited. Therefore, the SQH method, by providing accurate and precise cure rate estimates, plays a crucial role in treatment decisions and the development of effective adjuvant therapies, which are vital for improving patient survival.

In comparing algorithms, it is also essential to evaluate CPU times. For this, we used the general BCT cure model and present the CPU times (in seconds) in Table \ref{table:CPU}. Each reported time represents the total time taken by an algorithm to generate the estimation results (estimates, bias, and RMSE) over 500 Monte Carlo runs. We observe that the SQH algorithm performs faster for smaller sample sizes and higher cure rates, which contrasts with the behavior of the NCG algorithm. Overall, for all parameter settings considered, the SQH algorithm (being gradient-free) runs faster than the NCG algorithm. Notably, the SQH algorithm offers a significant reduction in runtime for smaller sample sizes compared to the NCG algorithm.


\begin{table} [htb!]
\caption{Comparison of SQH and NCG algorithms in terms of bias and RMSE for the BCT cure model with 
$\alpha=0.5$ and binary covariate}
\centering
\begin{tabular}{ r r r r r r r }\\ 
\hline                                 
$n$ & $(p_{01},p_{00})$ & Parameter &\multicolumn{2}{c}{Bias}  &\multicolumn{2}{c}{RMSE} \\  \cline{4-5}  \cline{6-7}
& & &SQH & NCG & SQH & NCG \\
\hline  

$150$ & $(0.40,0.20)$  & $\beta_0=0.905$                   & 0.089 &   0.102 &   0.103 &   0.141    \\  
 & &  $\beta_{1}=-0.755$                                    &  0.076 &   0.133 &   0.088 &   0.190    \\ 
 & &  $\gamma_1=0.316$                                       &  0.022 &   0.026 &   0.027 &   0.032   \\
 & &  $\gamma_2=0.179$                                       &  0.007  &  0.009 &   0.009 &   0.012 \\
 & &  $\alpha=0.500$                                          &  0.053 &   0.109 &   0.061 &   0.158  \\ [0.5ex] 
 
$200$ & $(0.40,0.20)$  & $\beta_0=0.905$                    & 0.093   & 0.097  &  0.106  &  0.118 \\
 & &  $\beta_{1}=-0.755$                                     & 0.075  &  0.097  &  0.087 &   0.131  \\ 
 & &  $\gamma_1=0.316$                                        &0.021   & 0.026  &  0.026 &   0.031 \\
 & &  $\gamma_2=0.179$                                        &0.007   & 0.011  &  0.009 &   0.015  \\
 & &  $\alpha=0.500$                                           &0.052   & 0.072  &  0.060 &   0.099 \\ [0.5ex]
 
 $300$ & $(0.40,0.20)$  & $\beta_0=0.905$                 &   0.090  &  0.094 &   0.105 &   0.109\\
 & &  $\beta_{1}=-0.755$                                   &  0.076  &  0.077  &  0.088  &  0.089 \\
 & &  $\gamma_1=0.316$                                     &  0.019  &  0.028  &  0.023  &  0.033 \\
 & &  $\gamma_2=0.179$                                     &  0.006  &  0.016  &  0.007  &  0.019 \\
 & &  $\alpha=0.500$                                       &  0.052  &  0.053  &  0.060  &  0.061  \\ [0.5ex]
 
 $150$ & $(0.65,0.35)$  & $\beta_0=0.322$            &    0.032 &   0.136 &   0.038 &   0.188  \\    
 & &  $\beta_{1}=-1.055$                             &    0.108  &  0.227 &   0.125 &   0.306   \\
 & &  $\gamma_1=0.316$                                &   0.024  &  0.032 &   0.029  &  0.040     \\ 
 & &  $\gamma_2=0.179$                                &    0.008 &   0.009 &    0.010 &   0.012 \\
 & &  $\alpha=0.500$                                  &   0.053  &  0.150  &  0.061   & 0.211    \\ [0.5ex] 
 
  $200$ & $(0.65,0.35)$  & $\beta_0=0.322$       &    0.032 &   0.114  &  0.037  &  0.162    \\      
 & &  $\beta_{1}=-1.055$                          &    0.110 &   0.179  &  0.126  &  0.243    \\     
 & &  $\gamma_1=0.316$                            &     0.023 &   0.027  &  0.028  &  0.034    \\    
 & &  $\gamma_2=0.179$                            &     0.007  &  0.009   & 0.009  &  0.012     \\   
 & &  $\alpha=0.500$                               &     0.055  &  0.121  &  0.062 &   0.174     \\ [0.5ex]  
 
  $300$ & $(0.65,0.35)$  & $\beta_0=0.322$    &     0.032  &  0.047   & 0.037  &  0.067     \\       
 & &  $\beta_{1}=-1.055$                      &      0.106 &   0.114  &  0.122  &  0.136      \\     
 & &  $\gamma_1=0.316$                        &      0.021 &   0.026  &  0.026  &  0.032       \\   
 & &  $\gamma_2=0.179$                        &       0.006 &   0.011 &   0.007 &   0.015       \\   
 & &  $\alpha=0.500$                          &     0.051   & 0.060   & 0.060   & 0.079     \\      
\hline
\end{tabular}
\label{table:T1}
\end{table}


\begin{table} [htb!]
\caption{Comparison of SQH and NCG algorithms in terms of bias and RMSE for the BCT cure model with 
$\alpha=0.75$ and binary covariate}
\centering
\begin{tabular}{ r r r r r r r }\\ 
\hline                                 
$n$ & $(p_{01},p_{00})$ & Parameter &\multicolumn{2}{c}{Bias}  &\multicolumn{2}{c}{RMSE} \\  \cline{4-5}  \cline{6-7}
& & &SQH & NCG & SQH & NCG \\
\hline  

$150$ & $(0.40,0.20)$  & $\beta_0=1.139$         &        0.115  &  0.121  &  0.132  &  0.155 \\
 & &  $\beta_{1}=-0.864$                         &        0.085  &  0.129  &  0.099  &  0.189 \\ 
 & &  $\gamma_1=0.316$                           &        0.023  &  0.027  &  0.028  &  0.033  \\ 
 & &  $\gamma_2=0.179$                           &        0.007  &  0.009  &  0.009  &  0.012  \\ 
 & &  $\alpha=0.750$                             &        0.078  &  0.107  &  0.090  &  0.140  \\ [0.5ex] 
 
$200$ & $(0.40,0.20)$  & $\beta_0=1.139$          &       0.113  &  0.115  &  0.130 &   0.138 \\
 & &  $\beta_{1}=-0.864$                          &       0.087  &  0.100  &  0.100 &   0.137  \\  
 & &  $\gamma_1=0.316$                            &       0.022  &  0.026  &  0.026 &   0.032   \\ 
 & &  $\gamma_2=0.179$                            &       0.007  &  0.012  &  0.008 &   0.016   \\ 
 & &  $\alpha=0.750$                              &       0.073  &  0.085  &  0.085 &   0.109   \\ [0.5ex]  
 
 $300$ & $(0.40,0.20)$  & $\beta_0=1.139$        &        0.113 &   0.116  &  0.130  &  0.133 \\
 & &  $\beta_{1}=-0.864$                         &        0.085 &   0.086  &  0.099  &  0.100 \\ 
 & &  $\gamma_1=0.316$                           &        0.018 &   0.029  &  0.023  &  0.034  \\
 & &  $\gamma_2=0.179$                           &        0.006 &   0.017  &  0.007  &  0.020  \\ 
 & &  $\alpha=0.750$                             &        0.074 &   0.076  &  0.086  &  0.090 \\ [0.5ex]  
 
 $150$ & $(0.65,0.35)$  & $\beta_0=0.468$       &        0.048  &  0.136 &   0.055  &  0.183  \\
 & &  $\beta_{1}=-1.144$                        &        0.114  &  0.218 &   0.132  &  0.297  \\
 & &  $\gamma_1=0.316$                          &      0.024    &  0.033 &   0.029  &  0.040  \\
 & &  $\gamma_2=0.179$                          &     0.008     &  0.009 &   0.010  &  0.011  \\
 & &  $\alpha=0.750$                            &     0.076     &  0.155 &   0.088  &  0.210 \\ [0.5ex] 
 
  $200$ & $(0.65,0.35)$  & $\beta_0=0.468$      &     0.047   &  0.119  &  0.054  &  0.160 \\
 & &  $\beta_{1}=-1.144$                        &      0.116  &  0.176  &  0.133  &  0.236 \\ 
 & &  $\gamma_1=0.316$                          &     0.023   &  0.028  &  0.027  &  0.034 \\
 & &  $\gamma_2=0.179$                          &    0.007    &  0.009  &  0.009  &  0.012 \\
 & &  $\alpha=0.750$                            &      0.076  &  0.126  &  0.088  &  0.174 \\
 
  $300$ & $(0.65,0.35)$  & $\beta_0=0.468$    &      0.046   &  0.058  &  0.053  &  0.076 \\
 & &  $\beta_{1}=-1.144$                      &       0.116  &  0.118  &  0.133  &  0.141 \\
 & &  $\gamma_1=0.316$                        &      0.021   &  0.027  &  0.025  &  0.033 \\
 & &  $\gamma_2=0.179$                        &       0.006  &  0.012  &  0.007  &  0.015 \\
 & &  $\alpha=0.750$                          &       0.074  &  0.079  &  0.086  &  0.096   \\
\hline
\end{tabular}
\label{table:T2}
\end{table} 


\begin{table} [htb!]
\caption{Comparison of SQH and NCG algorithms in terms of bias and RMSE for the BCT cure model with 
$\alpha$ on the boundary (i.e., $\alpha=0$ or 1) and binary covariate}
\centering
\begin{tabular}{ r r r r r r r }\\ 
\hline                                 
$n$ & $(p_{01},p_{00})$ & Parameter &\multicolumn{2}{c}{Bias}  &\multicolumn{2}{c}{RMSE} \\  \cline{4-5}  \cline{6-7}
& & &SQH & NCG & SQH & NCG \\
\hline  

$200$ & $(0.40,0.20)$  & $\beta_0=1.386$         &        0.140  &  0.138  &  0.161  &  0.162 \\
 & &  $\beta_{1}=-0.981$                         &       0.100   &  0.102  &  0.114  &  0.121 \\
 & &  $\gamma_1=0.316$                           &       0.022   &  0.028  &  0.026  &  0.034 \\
 & &  $\gamma_2=0.179$                           &        0.007  &  0.012  &  0.008  &  0.015 \\
 & &  $\alpha=1$                                 &      0.097    &  0.102  &  0.113  &  0.120 \\ [0.5ex]
 
$200$ & $(0.65,0.35)$  & $\beta_0=0.619$          &     0.065  &   0.091  &  0.074  &  0.119 \\
 & &  $\beta_{1}=-1.238$                          &      0.129 &   0.147  &  0.148  &  0.190 \\
 & &  $\gamma_1=0.316$                            &      0.022 &   0.028  &  0.027  &  0.034 \\ 
 & &  $\gamma_2=0.179$                            &     0.008  &   0.009  &  0.010  &  0.012 \\
 & &  $\alpha=1$                                  &     0.103  &   0.118  &  0.121  &  0.153 \\ [0.5ex]
 
 $200$ & $(0.40,0.20)$  & $\beta_0=0.476$        &        0.046  &  0.057  &  0.054  &  0.073 \\
 & &  $\beta_{1}=-0.563$                         &      0.055    &  0.073  &  0.064  &  0.096 \\
 & &  $\gamma_1=0.316$                           &      0.021    &  0.025  &  0.026  &  0.030 \\
 & &  $\gamma_2=0.179$                           &         0.007 &  0.011  &  0.008  &  0.015 \\
 & &  $\alpha=0$                                 &       0.013   &  0.029  &  0.015  &  0.065 \\ [0.5ex]
 
 $200$ & $(0.65,0.35)$  & $\beta_0=0.049$       &      0.008 &   0.077  &  0.010  &  0.128  \\
 & &  $\beta_{1}=-0.891$                        &      0.090 &   0.133  &  0.104  &  0.193  \\ 
 & &  $\gamma_1=0.316$                          &      0.022 &   0.027  &  0.027  &  0.033 \\
 & &  $\gamma_2=0.179$                          &    0.007   &   0.009  &  0.010  &  0.012 \\
 & &  $\alpha=0$                                &    0.010   &   0.064  &  0.011  &  0.145 \\

\hline
\end{tabular}
\label{table:T3}
\end{table} 


\begin{table} [htb!]
\caption{Comparison of SQH and NCG algorithms in terms of bias and RMSE for the BCT cure model with continuous covariate}
\centering
\begin{tabular}{ r r r r r r }\\ 
\hline                                 
$n$ &  Parameter &\multicolumn{2}{c}{Bias}  &\multicolumn{2}{c}{RMSE} \\  \cline{3-4}  \cline{5-6}
 & &SQH & NCG & SQH & NCG \\
\hline  

$150$   & $\beta_0=-0.746$                         &         0.075  &  0.138  &  0.086 &   0.215 \\
 &  $\beta_{1}=0.134$                              &         0.014  &  0.021  &  0.018 &   0.030 \\
  &  $\gamma_1=0.316$                              &         0.022  &  0.025  &  0.027 &   0.031 \\
  &  $\gamma_2=0.179$                              &         0.007  &  0.010  &  0.009 &   0.013 \\
  &  $\alpha=0.500$                                &         0.050  & 0.096   & 0.058  &  0.160  \\ [0.5ex]

 $300$   & $\beta_0=-0.746$                       &       0.071  &  0.073  &  0.082  &  0.085 \\
  &  $\beta_{1}=0.134$                            &       0.012  &  0.014  &  0.015  &  0.016 \\
  &  $\gamma_1=0.316$                             &       0.017  &  0.029  &  0.021  &  0.034 \\
  &  $\gamma_2=0.179$                             &       0.005  &  0.018  &  0.007  &  0.022 \\
  &  $\alpha=0.500$                               &       0.051  &  0.053  &  0.060  &  0.061 \\ [0.5ex]
 
 $150$  & $\beta_0=-0.692$                       &        0.069 &   0.117 &   0.080  &  0.201 \\
  &  $\beta_{1}=0.156$                           &        0.017 &   0.023 &   0.020  &  0.032 \\ 
  &  $\gamma_1=0.316$                            &        0.020 &   0.026 &   0.025  &  0.032 \\
  &  $\gamma_2=0.179$                            &        0.007 &   0.010 &   0.008  &  0.014 \\
  &  $\alpha=0.750$                              &        0.078 &   0.108 &   0.089  &  0.156 \\ [0.5ex]

  $300$   & $\beta_0=-0.692$                 &       0.069  &  0.071  &  0.080  &  0.082   \\
  &  $\beta_{1}=0.156$                       &        0.014 &   0.016 &   0.018 &   0.019 \\
  &  $\gamma_1=0.316$                        &        0.018 &   0.030 &   0.022 &   0.035 \\
  &  $\gamma_2=0.179$                        &        0.005 &   0.019 &   0.006 &   0.022  \\
  &  $\alpha=0.750$                          &        0.075 &   0.077 &   0.087 &   0.089 \\
\hline
\end{tabular}
\label{table:T4}
\end{table}


\begin{table} [htb!]
\caption{Comparison of SQH and NCG algorithms in terms of bias and RMSE of the estimate of $S_p(y|x)$ for the BCT cure model with $n=200$}
\centering
\setlength{\tabcolsep}{10pt}
\begin{tabular}{ r r r r | r r }\\ 
\hline                                 
$x$ & $(p_{01},p_{00},\alpha)$ &\multicolumn{2}{c}{Bias of $S_p(y=2|x)$}  &\multicolumn{2}{c}{RMSE of $S_p(y=2|x)$} \\  \cline{3-4}  \cline{5-6}
& &SQG & NCG & SQH & NCG \\
\hline  

$1$ & $(0.65,0.35,0.5)$        &    0.004    &   0.005   &    0.005   &    0.006 \\
 
$0$  & $(0.65,0.35,0.5)$       &    0.008   &    0.009   &     0.010   &    0.012 \\

$1$ & $(0.65,0.35,0.75)$      &     0.004   &    0.005  &     0.005  &     0.006 \\
 
$0$  & $(0.65,0.35,0.75)$     &    0.007  &     0.009   &    0.009    &   0.011 \\

$1$ & $(0.40,0.20,0.5)$      &    0.006   &    0.007   &    0.008   &    0.009 \\
 
$0$  & $(0.40,0.20,0.5)$     &     0.009   &    0.011   &    0.011   &    0.013 \\

$1$ & $(0.40,0.20,0.75)$     &    0.006    &   0.006   &    0.007  &     0.008 \\
 
$0$  & $(0.40,0.20,0.75)$    &    0.008   &    0.009   &     0.010   &    0.011 \\

\hline
\end{tabular}
\label{table:Sp}
\end{table}


\begin{table} [htb!]
\caption{Comparison of SQH and NCG algorithms in terms of bias and RMSE of the estimate of cure rate for the BCT cure model with binary covariate} 
\centering
\begin{tabular}{ r r r r r r r }\\ 
\hline                                 
$\alpha$ & $n$ & $p_0$ &\multicolumn{2}{c}{Bias}  &\multicolumn{2}{c}{RMSE} \\  \cline{4-5}  \cline{6-7}
& & &SQH & NCG & SQH & NCG \\
\hline  

$0.500$ & $200$  & $p_{01} = 0.400$        &          0.033  &  0.040  &  0.040 &   0.050    \\      
 & &  $p_{00} = 0.200$                     &          0.022  &  0.026  &  0.026 &   0.034     \\ [0.5ex]  
 
$0.500$ & $300$  & $p_{01} = 0.400$        &          0.032 &   0.034  &  0.039 &   0.042     \\ 
 & & $p_{00} = 0.200$                      &          0.021  &  0.022  &  0.025  &  0.027  \\ [0.5ex]

$0.750$ & $200$  & $p_{01} = 0.400$        &          0.035 &   0.039 &   0.044 &   0.049     \\     
 & &  $p_{00} = 0.200$                     &          0.023  &  0.025 &   0.027  &  0.032     \\ [0.5ex]   
 
$0.750$ & $300$  & $p_{01} = 0.400$        &          0.033  &  0.035 &   0.041  &  0.044   \\   
 & & $p_{00} = 0.200$                      &           0.022 &   0.024  &  0.027  &  0.029 \\ [0.5ex]

$0.500$ & $200$  & $p_{01} = 0.650$        &            0.028  &  0.045 &   0.032 &   0.056   \\ 
 & &  $p_{00} = 0.350$                     &            0.012 &   0.042  &  0.015  &  0.056    \\ [0.5ex]      
 
$0.500$ & $300$  & $p_{01} = 0.650$        &            0.028  &  0.031  &  0.033  &  0.037   \\     
 & &  $p_{00} = 0.350$                     &             0.012 &   0.017 &   0.015 &   0.025  \\ [0.5ex]

$0.750$ & $200$  & $p_{01} = 0.650$        &             0.029  &  0.044  &  0.034  &  0.055    \\ 
 & &  $p_{00} = 0.350$                     &              0.016 &   0.040 &   0.019 &   0.053   \\ [0.5ex]      
 
$0.750$ & $300$  & $p_{01} = 0.650$        &             0.029   & 0.031  &  0.034  &  0.038    \\    
 & &  $p_{00} = 0.350$                     &              0.016 &   0.019  &  0.019  &  0.025  \\
\hline
\end{tabular}
\label{table:cure}
\end{table}


\begin{table} [htb!]
\caption{Comparison of SQH and NCG algorithms in terms of CPU times (in seconds)}
\renewcommand{\arraystretch}{1.1}
\centering
\begin{tabular}{ r r r r}\\ 
\hline                                 
$n$ & $(\alpha,p_{01},p_{00})$ &\multicolumn{2}{c}{CPU Time (seconds)} \\  \cline{3-4} 
& & SQH & NCG \\
\hline  
$200$ & $(0.500,0.400,0.200)$ &    3.310        &   30.717         \\[0.5ex]
$300$ & $(0.500,0.400,0.200)$ &    4.468       &     7.318       \\[0.5ex]
$200$ & $(0.750,0.400,0.200)$ &    3.124        &     24.995       \\[0.5ex]
$300$ & $(0.750,0.400,0.200)$ &    4.414         &    6.898        \\[0.5ex]

$200$ & $(0.500,0.650,0.350)$ &     2.996       &    86.168        \\[0.5ex]
$300$ & $(0.500,0.650,0.350)$ &     3.896       &    31.186        \\[0.5ex]
$200$ & $(0.750,0.650,0.350)$ &     2.961       &    79.945        \\[0.5ex]
$300$ & $(0.750,0.650,0.350)$ &     3.833       &    23.534        \\[0.5ex]

\hline
\end{tabular}
\label{table:CPU}
\end{table}

\subsubsection{Discussion on the fixed parameters in the SQH algorithm}

As mentioned earlier, the fixed parameters in the SQH algorithm ($\epsilon$, $\lambda$, $\rho$, and $\zeta$) were selected after a careful preliminary study, with all studies based on 500 Monte Carlo runs. We observed that the estimation results remained similar for different values of $\zeta \in (0,1)$; see Table \ref{table:Sen} for a sensitivity analysis with respect to $\zeta$ when $\epsilon$, $\lambda$ and $\rho$ are all fixed at 1000. However, we found that the CPU time decreased as $\zeta$ increased. For example, with the setting ($n=200, p_{01}=0.65, p_{00}=0.35, \alpha=0.5$), the CPU times (in seconds) for $\zeta=0.1$, $\zeta=0.5$, and $\zeta=0.9$ were 3.497, 3.042, and 2.939, respectively. We also conducted a sensitivity analysis for $\lambda (> 1)$, keeping $\epsilon$, $\rho$, and $\zeta$ fixed at 1000, 1000, and 0.5, respectively. With the setting  ($n=200, p_{01}=0.40, p_{00}=0.20, \alpha=0.5$), we varied $\lambda$ within the range $[1.1,1000]$ and observed that the estimation results remained consistent. CPU times for different values of $\lambda$ were similar, with $\lambda=1.1$ resulting in 3.348 seconds and $\lambda=1000$ yielding the smallest CPU time of 3.262 seconds. Next, we performed a sensitivity analysis for $\rho (>0)$, keeping $\epsilon$, $\lambda$, and $\zeta$ fixed at 1000, 1000, and 0.5, respectively. With the setting ($n=200, p_{01}=0.40, p_{00}=0.20, \alpha=0.5$) and $\rho$ varying within $[0.1,1000]$, we observed very minor variation in the parameter estimates. The CPU times ranged from 3.282 to 3.438 seconds. Finally, we conducted a sensitivity analysis for $\epsilon (> 0)$, keeping $\lambda$, $\rho$, and $\zeta$ fixed at 1000, 1000, and 0.5, respectively, and varying $\epsilon$ within $[0.1,1000]$. For $\epsilon \in (0,40)$ and with the setting ($n=200, p_{01}=0.40, p_{00}=0.20, \alpha=0.5$), we observed some divergent cases, with the number of divergent samples decreasing as $\epsilon$ increased, and no divergences occurred for $\epsilon\geq 40$. The estimation results were similar for all values of $\epsilon$, and the CPU time ranged from 3.349 to 5.156 seconds, with $\epsilon=1000$ yielding the smallest CPU time.

\begin{table} [htb!]
\caption{Sensitivity analysis of the SQH algorithm with respect to the parameter $\zeta \in (0,1)$}
\centering
\begin{tabular}{r r r r r | r r | r r}\\ 
\hline                                 
$n$ & $(p_{01},p_{00})$ &  Parameter &\multicolumn{2}{c}{$\zeta=0.1$}  &\multicolumn{2}{c}{$\zeta=0.5$} &\multicolumn{2}{c}{$\zeta=0.9$} \\  \cline{4-5}  \cline{6-7} \cline{8-9}
& & &Bias & RMSE & Bias & RMSE & Bias & RMSE  \\
\hline  

200 & $(0.65,0.35)$   & $\beta_0=0.322$      &      0.032 &   0.038  &   0.032  &  0.037  &   0.032  &  0.037 \\
  &                   &  $\beta_{1}=-1.055$    &      0.110 &   0.127  &   0.110  &  0.126  &   0.110  &  0.126 \\
  &                   &  $\gamma_1=0.316$     &      0.023 &   0.028  &   0.023  &  0.028  &   0.022  &  0.027 \\
  &                   &  $\gamma_2=0.179$     &      0.007 &   0.009  &   0.007  &  0.009  &   0.007  &  0.009 \\
  &                   &  $\alpha=0.500$       &      0.055 &   0.062  &   0.055  &  0.062  &   0.055  &  0.062 \\ [0.75ex] \hline

200 & $(0.65,0.35)$   & $\beta_0=0.468$       &     0.047 &   0.055  &   0.047  &  0.054   &    0.047  &  0.055 \\
  &                   &  $\beta_{1}=-1.144$   &     0.117 &   0.134  &   0.116  &  0.133   &    0.116  &  0.133 \\
  &                   &  $\gamma_1=0.316$     &     0.023 &   0.028  &   0.023  &  0.027   &    0.022  &  0.027 \\
  &                   &  $\gamma_2=0.179$     &     0.007 &   0.009  &   0.007  &  0.009   &    0.007  &  0.009 \\
  &                   &  $\alpha=0.750$       &     0.076 &   0.088  &   0.076  &  0.088   &    0.076  &  0.088 \\ [0.75ex] \hline

  200 & $(0.40,0.20)$   & $\beta_0=0.905$    &      0.094 &   0.107  &  0.093  &  0.106  &   0.094  &  0.107 \\
  &                   &  $\beta_{1}=-0.755$    &     0.075 &   0.087  &  0.075  &  0.087  &   0.075  &  0.087 \\
  &                   &  $\gamma_1=0.316$     &      0.022 &   0.027  &  0.021  &  0.026  &   0.021  &  0.026 \\
  &                   &  $\gamma_2=0.179$     &      0.007 &   0.008  &  0.007  &  0.009  &   0.007  &  0.009 \\
  &                   &  $\alpha=0.500$       &      0.052 &   0.060  &  0.052  &  0.060  &   0.052  &  0.060 \\ [0.75ex] \hline

  200 & $(0.40,0.20)$   & $\beta_0=1.139$     &      0.114  &  0.131  &  0.113  &  0.130  &   0.113  &  0.130 \\
  &                   &  $\beta_{1}=-0.864$   &      0.087  &  0.100  &  0.087  &  0.100  &   0.087  &  0.100 \\
  &                   &  $\gamma_1=0.316$     &      0.022  &  0.027  &  0.022  &  0.026  &   0.021  &  0.026 \\
  &                   &  $\gamma_2=0.179$     &      0.006  &  0.008  &  0.007  &  0.008  &   0.007  &  0.008 \\
  &                   &  $\alpha=0.750$       &      0.073  &  0.086  &  0.073  &  0.085  &   0.073  &  0.085 \\

\hline
\end{tabular}
\label{table:Sen}
\end{table}

\section{Application: analysis of cutaneous melanoma data}

To demonstrate the application of the proposed SQH algorithm in the BCT cure rate model, we used the well-known cutaneous melanoma dataset, which was also employed by \cite{PalRoy23} to illustrate the NCG algorithm. This dataset is part of an assay on cutaneous melanoma, a type of malignant cancer, used to evaluate the performance of postoperative treatment with a high dose of interferon alpha-2b. After excluding patients with missing covariate data, the dataset contains complete information on 417 patients observed between 1991 and 1995, who were followed until 1998. The data includes 56\% right-censored observations. The observed time, defined as the time (in years) until death or censoring, has a mean of 3.18 years and a standard deviation of 1.69 years. Following \cite{PalRoy23}, we considered the nodule category as the sole covariate $(x)$, with categories: $1: n_1=111, 2: n_2 = 137, 3: n_3 = 87,$ and $4: n_4 = 82$.

To compute the initial values of the model parameters for starting the SQH algorithm, we first calculated the non-parametric estimates of the cured proportions for nodule categories 1 and 4 by using the Kaplan-Meier estimate of the survival function at the largest observed lifetime. These estimates were then equated to $p_0(x)$ with $x=1$ and 4, as described in \eqref{p0}, for different values of $\alpha$ in the interval $[0,1]$. By solving the resulting equations for different values of $\alpha$, we obtained a set of values for the regression parameters $\beta_0$ and $\beta_1$. For each pair $(\beta_0,\beta_1)$ and its corresponding value of $\alpha$, we calculated the log-likelihood function value. To do this, we equated the mean and variance of the Weibull density function in \eqref{wei} to the mean and variance of the observed data to determine the parameters $\gamma_1$ and $\gamma_2$. Finally, we selected the set of model parameters that maximized the log-likelihood value as our initial guess.

Using the initial values derived from the technique described above, a Weibull distribution fitted for the lifetime as in \eqref{wei}, and the SQH parameters $\epsilon$, $\lambda$, $\rho$, and $\zeta$ set to 1000, 1000, 1000, and 0.5, respectively, the SQH estimates obtained were $\hat{\beta_0} = -1.228$, $\hat{\beta_1} = 0.386$, $\hat{\gamma_1} = 0.561$, $\hat{\gamma_2} = 0.376$, and $\hat{\alpha} = 0.051$. These estimates are consistent with those reported in \cite{PalRoy23}. To compute the standard errors, a bootstrap method was employed as recommended by \cite{PalRoy23}. The non-parametric bootstrap standard error estimates, based on a sample size of 500, were $\text{s.e}(\hat{\beta_0}) = 0.115$, $\text{s.e}(\hat{\beta_1}) = 0.024$, $\text{s.e}(\hat{\gamma_1}) = 0.022$, $\text{s.e}(\hat{\gamma_2}) = 0.024$, and $\text{s.e}(\hat{\alpha}) = 0.008$. Note that the SQH-based standard errors are smaller compared to the NCG-based standard errors reported by \cite{PalRoy23}. With these SQH parameter estimates, the calculated cure rates for the four nodule categories were $\hat{p}_{01}=0.653$, $\hat{p}_{02}=0.536$, $\hat{p}_{03}=0.402$, and $\hat{p}_{04}=0.266$. These estimates align with the findings of \cite{Bal16} and \cite{PalRoy23}.

In Figure \ref{figure:SP}, we present the estimated population survival function, stratified by nodule category, and overlay it with the corresponding non-parametric Kaplan-Meier survival function estimates. To assess the adequacy of the BCT cure model with Weibull lifetimes using the SQH estimates, we calculated the normalized randomized quantile residuals. Figure \ref{figure:QQ} displays the QQ plot, where each point represents the median of five sets of ordered residuals. The plot indicates that the BCT cure model with Weibull lifetimes provides a good fit to the melanoma data. Additionally, we formally tested the normality of the residuals using the Kolmogorov-Smirnov test, which yielded a p-value of 0.933, offering strong evidence in favor of the normality of the residuals.
\begin{figure}[htb!]
\centering
\includegraphics[width=0.6\linewidth]{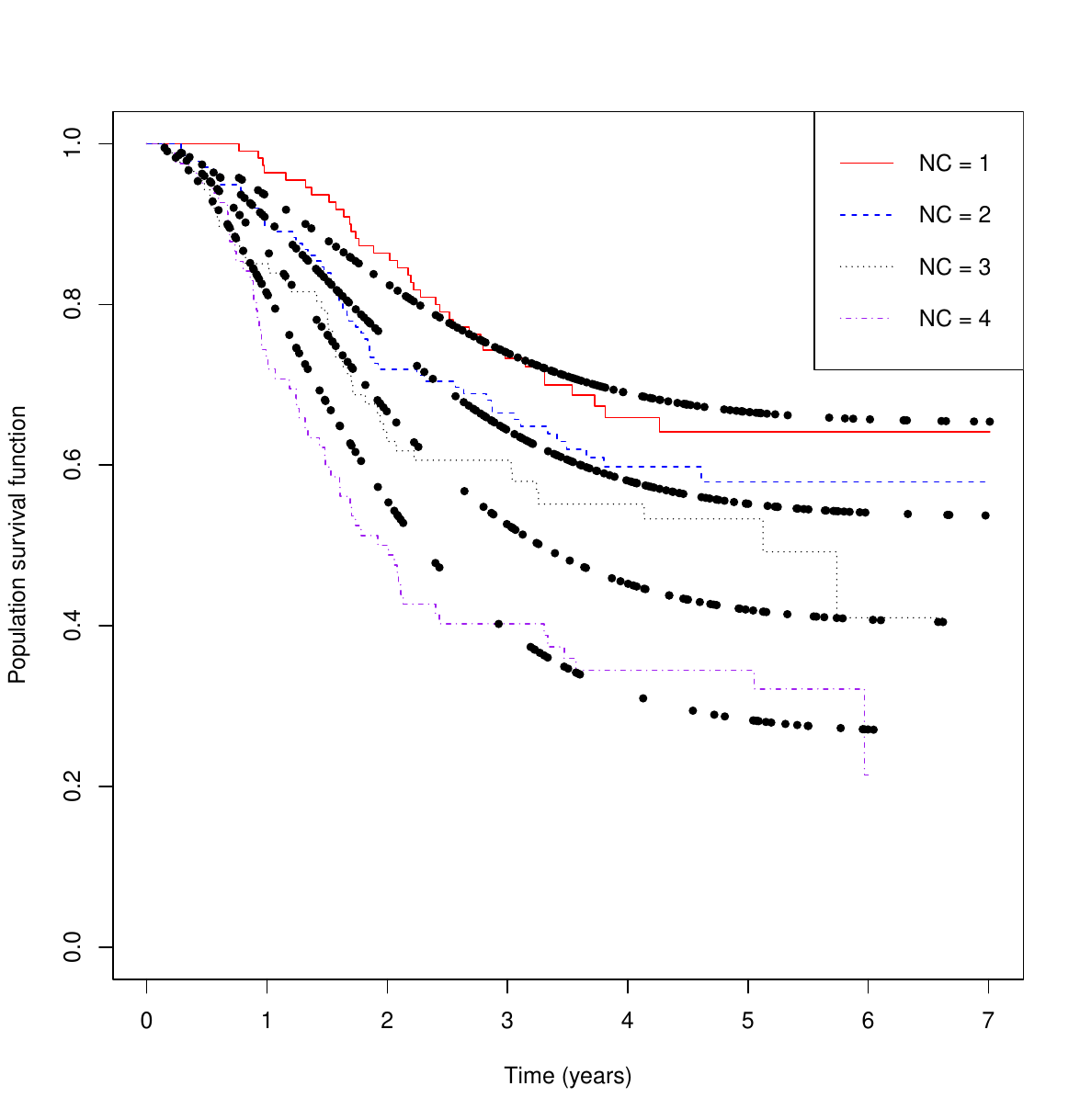}
\caption{Plots of the estimated population survival function stratified by nodule category (NC)}
\label{figure:SP}
\end{figure}
\begin{figure}[htb!]
\centering
\includegraphics[width=0.5\linewidth]{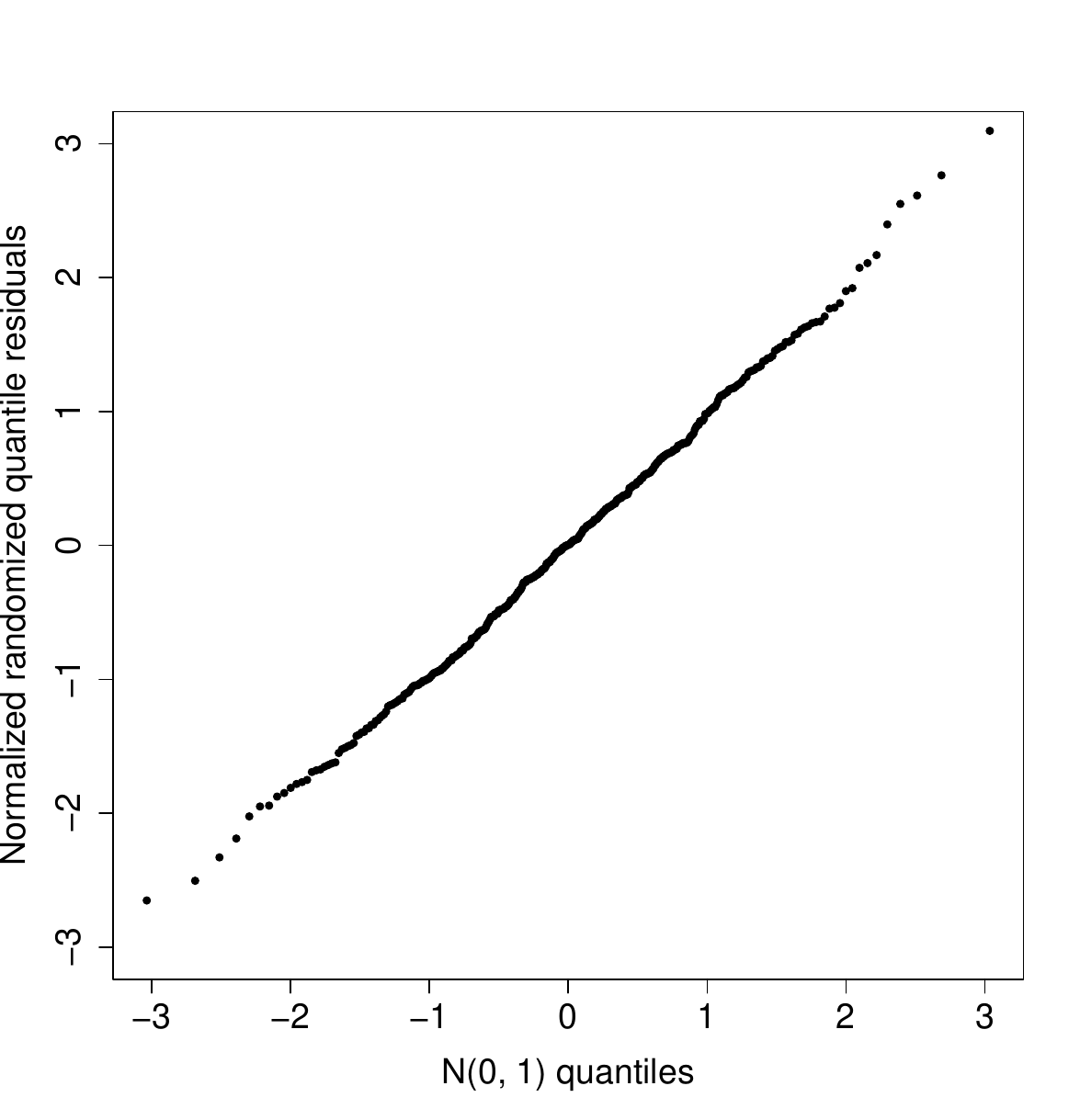}
\caption{QQ plot of normalized randomized quantile residuals}
\label{figure:QQ}
\end{figure}

\section{Conclusion and future work}

In this paper, we focused on parameter estimation for the BCT cure model. We proposed an improved estimation method, the SQH algorithm, which is based on a gradient-free technique. While we applied the SQH algorithm to the maximum likelihood estimation of the BCT cure model parameters, it can be used for other maximum likelihood estimation problems. Compared to the recently developed NCG algorithm (which has already been shown to outperform the EM algorithm), the SQH algorithm yields lower bias and RMSE for all model parameters, leading to more accurate estimates of the cure rates. We also compared CPU times and found that the SQH algorithm runs faster than the NCG algorithm across all parameter settings considered in our simulation study. In the analysis of the cutaneous melanoma data, we observed that the SQH-based standard errors are smaller compared to the NCG-based standard errors.

The SQH algorithm could be valuable for other complex cure rate models, such as the Conway-Maxwell (COM) Poisson cure rate model and the destructive COM-Poisson cure rate model \citep{Bal16,Majak19}, where a profile likelihood approach has been suggested for estimating the COM-Poisson shape parameter $\phi$ due to the flatness of the likelihood surface with respect to $\phi$. It would be interesting to evaluate the performance of the SQH algorithm in these cases and compare it with the NCG algorithm, the EM algorithm, or its variations \citep{Davies21,Pal21SIM}. We are currently investigating these issues and plan to present our findings in a future paper.

\section*{Funding}

The research reported in this publication was supported by the National Institute Of General Medical Sciences of the National Institutes of Health under Award Number R15GM150091. The content is solely the responsibility of the authors and does not necessarily represent the official views of the National Institutes of Health. The work of S. Pal and S. Roy was also partly supported by the US National Science Foundation grant number 2309491.

\section*{Conflict of interest}

The authors declare no potential conflict of interest.

\section*{Data Availability Statement}

The R codes for the data generation and the SQH algorithm are available in the GitHub page \url{https://github.com/suvrapal/SQH-Box-Cox}.

\bibliographystyle{apacite}

\bibliography{Ref}%

\end{document}